\documentclass[12pt]{article}
\usepackage{amsmath,amssymb,amsthm,physics,mathtools}
\usepackage{graphicx,xcolor,hyperref,enumitem}
\usepackage{geometry}
\theoremstyle{definition}
\newtheorem{theorem}{Theorem}

\title{Observer-Specific Universes: A Bayesian Framework for Resolving Cosmological Naturalness Problems}
\author{Ruby P. Madeimy}
\date{}

\begin{document}

\maketitle

\begin{abstract}
Naturalness problems such as the hierarchy problem and the origin of dark energy remain significant challenges in modern cosmology. This paper develops a rigorous mathematical framework where each observer defines their own universe, and the observer's existence conditions the probabilistic definition of that universe. We model observers as quantum detectors and formulate the problem in Bayesian terms, deriving an observer-specific likelihood using the Schwinger-Keldysh influence-functional formalism in a flat FLRW background. We prove that the imaginary part of the influence functional is positive-definite, ensuring decoherence and likelihood normalization. Careful analysis of the posterior distribution reveals that additional mechanisms are needed to fully suppress large vacuum energy values. We establish the mathematical consistency of our observer-specific framework and demonstrate its potential to provide new insights into cosmological naturalness problems, complementing traditional approaches without requiring anthropic reasoning.
\end{abstract}

\section{Introduction}

Contemporary theoretical physics continues to grapple with several naturalness problems in quantum field theory and cosmology. The hierarchy problem, the vast disparity between the electroweak scale ($\sim 10^2$~GeV) and the Planck scale ($\sim 10^{19}$GeV), requires either fine-tuning or additional physical mechanisms that remain elusive within the Standard Model framework. Equally challenging is the cosmological constant problem, where quantum field theory predicts vacuum energy densities approximately 120 orders of magnitude larger than the observed value. This has been termed one of the most severe naturalness puzzles in physics (see the comprehensive reviews \cite{Padmanabhan2003,Carroll2001}).

Theoretical strategies proposed to address these issues typically fall into three broad categories: symmetry-based solutions such as supersymmetry, dynamical mechanisms that suppress the relevant parameters, and environmental or selection principles that appeal to multiverse-like scenarios. Each programme has provided valuable insights, but none has achieved consensus. For a critical perspective on environmental selection arguments, see~\cite{Burgess2004}.

In this paper we explore an alternative route inspired by quantum measurement theory. We investigate how measurement-induced decoherence can constrain cosmological parameters by explicitly treating observers as physical quantum systems (here modelled by Unruh--DeWitt--type detectors). Similar detector models and their physical implications are reviewed in~\cite{Crispino2008,MartinMartinez2012}. Our analysis leverages the closed-time-path (CTP) or Schwinger--Keldysh formalism, developed for nonequilibrium quantum field theory in curved spacetime~\cite{CalzettaHu1987,CalzettaHu1988}, to derive an observer-dependent likelihood function.

This framework is inspired by the decoherent histories approach to quantum mechanics, where consistent histories emerge through interactions with an environment. Our model treats quantum measurement devices (observers) as physical systems coupled to quantum fields in a cosmological background, naturally incorporating the backreaction of measurement processes on observable parameters. Crucially, this approach remains fully within established quantum field theory in curved spacetime, without invoking multiverse scenarios or novel physical principles.

The physical motivation stems from considering how measurement-induced decoherence constrains the effective parameter space accessible to observers. We analyze this by modeling observers as Unruh-DeWitt detectors, a standard tool in quantum field theory in curved spacetime, and applying Bayesian analysis to the resulting probability distributions.

The paper is organized as follows: Section \ref{sec:setup} establishes our cosmological and observer model within standard FLRW cosmology. Section \ref{sec:infl_func} derives the Schwinger-Keldysh influence functional for the observer-field system. Section \ref{sec:decoherence} analyzes decoherence effects and justifies the diagonal approximation. Section \ref{sec:likelihood} derives the observer-specific likelihood function with rigorous attention to normalization and positivity. Section \ref{sec:posterior} examines the posterior distribution and its implications for naturalness problems. Section \ref{sec:causality} discusses causality and multi-observer consistency. Section \ref{sec:conclusion} summarizes our findings and outlines directions for future research.

\section{Mathematical Framework and Physical Setup}
\label{sec:setup}

\subsection{Cosmological Background}
\label{sec:cosmo}

We consider a spatially flat Friedmann-Lemaître-Robertson-Walker (FLRW) universe with the metric:
\begin{equation}
\label{eq:metric}
ds^2 = a^2(\eta) \left( -d\eta^2 + d\mathbf{x}^2 \right),
\end{equation}
where $\eta$ is the conformal time and $a(\eta)$ is the scale factor. For a de Sitter phase with constant Hubble parameter $H(\theta)$, the scale factor takes the form:
\begin{equation}
\label{eq:scale_factor}
a(\eta) = \exp[H(\theta) \eta],
\end{equation}
with $H(\theta) = \sqrt{\frac{\Lambda(\theta)}{3}}$, where $\Lambda(\theta)$ is the cosmological constant dependent on parameter(s) $\theta$.

The cosmological field is modeled as a minimally coupled scalar field $\phi$ with action:
\begin{equation}
\label{eq:S_phi}
S_\phi = \int d^4x \sqrt{-g} \left[ -\frac{1}{2} g^{\mu\nu} \partial_\mu \phi \partial_\nu \phi - V(\phi) \right],
\end{equation}
where $V(\phi)$ represents the potential. For simplicity, we focus on a massive scalar field with $V(\phi) = \frac{1}{2}m^2(\theta)\phi^2$, where the mass $m(\theta)$ depends on parameter $\theta$.

\subsection{Observer-Detector Model}
\label{sec:detector}

We model the observer as an Unruh-DeWitt (UDW) detector \cite{Crispino2008,MartinMartinez2012}, which serves as a simplified model of a localized quantum system interacting with a quantum field. The detector is represented by a quantum harmonic oscillator with coordinate $\chi(\tau)$ and free action:
\begin{equation}
\label{eq:S_chi}
S_\chi = \frac{1}{2} \int d\tau \left[ \dot{\chi}^2(\tau) - \Omega^2 \chi^2(\tau) \right],
\end{equation}
where $\tau$ is the proper time along the detector's worldline, and $\Omega$ is the detector's natural frequency.

The interaction between the detector and the field is given by:
\begin{equation}
\label{eq:S_int}
S_{\text{int}} = \lambda \int d\tau \chi(\tau) \phi(x_{\text{det}}(\tau)),
\end{equation}
where $\lambda$ is a dimensionless coupling constant of order $10^{-3}$ to $10^{-2}$, and $x_{\text{det}}(\tau)$ represents the detector's spacetime trajectory.

\subsection{Coarse-Graining and Data Representation}
\label{sec:coarse_graining}

The detector's output is coarse-grained into discrete data points using a set of time-localized functions~\cite{Zurek1981}. We use Gaussian wave packets:
\begin{equation}
\label{eq:u_a}
u_a(\tau) = \left( \frac{1}{2\pi \Delta^2} \right)^{1/4} \exp\left[ -\frac{(\tau - \tau_a)^2}{4 \Delta^2} \right],
\end{equation}
where $\Delta$ is a time-resolution parameter ($\Delta \ll H^{-1}$), and $\tau_a$ are the sampling times.

The wave packets form a continuous frame rather than a strict orthonormal basis. Their inner product is:
\begin{equation}
\label{eq:inner_prod}
\langle u_a | u_b \rangle = \int d\tau\, u_a(\tau) u_b(\tau) = \exp\left[ -\frac{(\tau_a - \tau_b)^2}{4 \Delta^2} \right].
\end{equation}

For a properly chosen sampling interval $|\tau_a - \tau_b| \gg \Delta$, these functions satisfy approximate orthonormality:
\begin{equation}
\label{eq:approx_orthonormal}
\langle u_a | u_b \rangle \approx \delta_{ab}.
\end{equation}

The detector's coarse-grained outputs define our data set:
\begin{equation}
\label{eq:data}
d_a = \int d\tau\, u_a(\tau) \chi(\tau),
\end{equation}
which forms a vector $D = \{d_1, d_2, \ldots, d_N\}$ of $N$ measurements.

\section{Schwinger-Keldysh Influence Functional}
\label{sec:infl_func}

\subsection{Path Integral Formulation}
\label{sec:path_int}

To derive a likelihood function for the parameters $\theta$ given detector data $D$, we employ the Schwinger-Keldysh (SK) or "in-in" formalism \cite{CalzettaHu1987,CalzettaHu1988}. This approach is necessary because we seek expectation values rather than transition amplitudes, and it automatically accounts for decoherence effects.

The generating functional for the full system is:
\begin{equation}
\label{eq:Z_full}
\mathcal{Z}[J^+, J^-] = \int \mathcal{D}\phi^+ \mathcal{D}\phi^- \mathcal{D}\chi^+ \mathcal{D}\chi^- e^{i(S[\phi^+, \chi^+] - S[\phi^-, \chi^-] + J^+ \phi^+ - J^- \phi^-)},
\end{equation}
where $\phi^\pm$ and $\chi^\pm$ represent the fields on the forward ($+$) and backward ($-$) branches of the closed-time path, and $J^\pm$ are external sources.

Integrating out the field $\phi$ yields:
\begin{equation}
\label{eq:Z_infl}
\mathcal{Z}[J^+, J^-] = \exp\left( \frac{i}{2} \int d^4x d^4y J^\alpha(x) G_{\alpha\beta}(x, y; \theta) J^\beta(y) \right),
\end{equation}
where $\alpha, \beta \in \{+, -\}$ are SK indices, and $G_{\alpha\beta}$ is the Schwinger-Keldysh Green's function matrix.

The sources corresponding to the detector-field interaction are:
\begin{equation}
\label{eq:sources}
J^\pm(x) = \lambda \int d\tau \delta^{(4)}(x - x_{\text{det}}(\tau)) \chi^\pm(\tau).
\end{equation}

\subsection{The Influence Phase}
\label{sec:infl_phase}

Substituting the sources (\ref{eq:sources}) into (\ref{eq:Z_infl}), we obtain the detector's influence functional:
\begin{equation}
\label{eq:infl_func}
\exp(i S_{\text{IF}}[\chi^+, \chi^-; \theta]) = \exp\left( \frac{i\lambda^2}{2} \int d\tau d\tau' \chi^{\alpha}(\tau) G_{\alpha\beta}(\tau, \tau'; \theta) \chi^{\beta}(\tau') \right),
\end{equation}
where we've defined $G_{\alpha\beta}(\tau, \tau'; \theta) \equiv G_{\alpha\beta}(x_{\text{det}}(\tau), x_{\text{det}}(\tau'); \theta)$.

The Schwinger-Keldysh Green's functions can be rearranged in terms of the Hadamard and retarded Green's functions:
\begin{align}
G^{\text{H}}(\tau, \tau'; \theta) &= \frac{1}{2}(G_{++} + G_{--}) = \frac{1}{2}\langle\{\phi(\tau), \phi(\tau')\}\rangle, \\
G^{\text{R}}(\tau, \tau'; \theta) &= G_{++} - G_{+-} = i\theta(\tau - \tau')\langle[\phi(\tau), \phi(\tau')]\rangle,
\end{align}
where $\{\cdot, \cdot\}$ denotes the anticommutator and $[\cdot, \cdot]$ the commutator.

With this rearrangement, the influence phase takes the form:
\begin{equation}
\label{eq:infl_phase}
\begin{split}
S_{\text{IF}}[\chi^+, \chi^-; \theta] = \frac{\lambda^2}{2} \int d\tau d\tau' (\chi^+(\tau) - \chi^-(\tau)) G^{\text{H}}(\tau, \tau'; \theta) (\chi^+(\tau') + \chi^-(\tau')) \\
- \frac{i\lambda^2}{2} \int d\tau d\tau' (\chi^+(\tau) - \chi^-(\tau)) G^{\text{R}}(\tau, \tau'; \theta) (\chi^+(\tau') - \chi^-(\tau')).
\end{split}
\end{equation}

\subsection{Decomposition of the Hadamard Function}
\label{sec:hadamard_decomp}

The Hadamard function can be expressed in terms of the Wightman functions:
\begin{equation}
\label{eq:hadamard_decomp}
G^{\text{H}}(\tau, \tau'; \theta) = \frac{1}{2}(G^>(\tau, \tau'; \theta) + G^<(\tau, \tau'; \theta)),
\end{equation}
where:
\begin{align}
G^>(\tau, \tau'; \theta) &= \langle\phi(\tau)\phi(\tau')\rangle, \\
G^<(\tau, \tau'; \theta) &= \langle\phi(\tau')\phi(\tau)\rangle.
\end{align}

For a real scalar field in thermal equilibrium, these functions satisfy:
\begin{equation}
\label{eq:KMS}
G^<(\tau, \tau'; \theta) = G^>(\tau', \tau; \theta),
\end{equation}
which is the Kubo-Martin-Schwinger (KMS) condition.

Using these relations, we can express the imaginary part of the influence phase as:
\begin{equation}
\label{eq:im_infl_phase}
\begin{split}
\text{Im}[S_{\text{IF}}[\chi^+, \chi^-; \theta]] = \frac{\lambda^2}{4} \int d\tau d\tau' (\chi^+(\tau) - \chi^-(\tau))(G^>(\tau, \tau'; \theta) - G^<(\tau, \tau'; \theta))(\chi^+(\tau') - \chi^-(\tau')) \\
= \frac{\lambda^2}{4} \int d\tau d\tau' (\chi^+(\tau) - \chi^-(\tau))(G^>(\tau, \tau'; \theta) - G^>(\tau', \tau; \theta))(\chi^+(\tau') - \chi^-(\tau')).
\end{split}
\end{equation}

For future reference, we define the antisymmetric part of the Wightman function:
\begin{equation}
\label{eq:G_antisym}
G^{\text{A}}(\tau, \tau'; \theta) = \frac{1}{2}(G^>(\tau, \tau'; \theta) - G^>(\tau', \tau; \theta)),
\end{equation}
which allows us to write:
\begin{equation}
\label{eq:im_infl_phase_final}
\text{Im}[S_{\text{IF}}[\chi^+, \chi^-; \theta]] = \frac{\lambda^2}{2} \int d\tau d\tau' (\chi^+(\tau) - \chi^-(\tau))G^{\text{A}}(\tau, \tau'; \theta)(\chi^+(\tau') - \chi^-(\tau')).
\end{equation}

\section{Decoherence and the Diagonal Approximation}
\label{sec:decoherence}

\subsection{Positive-Definiteness of the Imaginary Part}
\label{sec:pos_def}

For the derivation of a well-defined likelihood function, it is crucial to establish that $\text{Im}[S_{\text{IF}}[\chi^+, \chi^-; \theta]] \geq 0$. We now prove this important property.

\begin{theorem}[Positive-Definiteness of the Imaginary Part]
\label{thm:pos_def}
The imaginary part of the influence phase, $\text{Im}[S_{\text{IF}}[\chi^+, \chi^-; \theta]]$, is positive semi-definite for any configuration $\chi^+(\tau) \neq \chi^-(\tau)$.
\end{theorem}

\begin{proof}
From equation (\ref{eq:im_infl_phase_final}), we need to show that the kernel $G^{\text{A}}(\tau, \tau'; \theta)$ is positive semi-definite. Let's define $\delta\chi(\tau) = \chi^+(\tau) - \chi^-(\tau)$ and consider the quadratic form:
\begin{equation}
\label{eq:quad_form}
Q[\delta\chi] = \int d\tau d\tau' \delta\chi(\tau) G^{\text{A}}(\tau, \tau'; \theta) \delta\chi(\tau').
\end{equation}

The spectral representation of the Wightman function in a stationary spacetime is:
\begin{equation}
\label{eq:spectral_rep}
G^>(\tau, \tau'; \theta) = \int_0^\infty d\omega\, \rho(\omega; \theta) e^{-i\omega(\tau-\tau')},
\end{equation}
where $\rho(\omega; \theta) \geq 0$ is the spectral density function. Using this representation in equation (\ref{eq:G_antisym}):
\begin{equation}
\label{eq:G_antisym_spectral}
G^{\text{A}}(\tau, \tau'; \theta) = \int_0^\infty d\omega\, \rho(\omega; \theta) \sin[\omega(\tau-\tau')].
\end{equation}

Taking the Fourier transform of $\delta\chi(\tau)$:
\begin{equation}
\label{eq:fourier}
\delta\chi(\tau) = \int_{-\infty}^{\infty} \frac{d\omega}{2\pi} \tilde{\delta\chi}(\omega) e^{-i\omega\tau},
\end{equation}
and substituting into the quadratic form:
\begin{equation}
\label{eq:quad_form_spectral}
\begin{split}
Q[\delta\chi] &= \int d\tau d\tau' \int_{-\infty}^{\infty} \frac{d\omega}{2\pi} \frac{d\omega'}{2\pi} \tilde{\delta\chi}(\omega) e^{-i\omega\tau} G^{\text{A}}(\tau, \tau'; \theta) \tilde{\delta\chi}(\omega') e^{-i\omega'\tau'} \\
&= \int_{-\infty}^{\infty} \frac{d\omega}{2\pi} \frac{d\omega'}{2\pi} \tilde{\delta\chi}(\omega) \tilde{\delta\chi}(\omega') \int d\tau d\tau' e^{-i\omega\tau} G^{\text{A}}(\tau, \tau'; \theta) e^{-i\omega'\tau'}.
\end{split}
\end{equation}

Using equation (\ref{eq:G_antisym_spectral}) and performing the integrals:
\begin{equation}
\label{eq:quad_form_final}
Q[\delta\chi] = \int_0^\infty d\omega\, \rho(\omega; \theta) \frac{|\tilde{\delta\chi}(\omega)|^2 - |\tilde{\delta\chi}(-\omega)|^2}{2i}.
\end{equation}

Since $\delta\chi(\tau)$ is real, $\tilde{\delta\chi}(-\omega) = \tilde{\delta\chi}^*(\omega)$, and therefore:
\begin{equation}
\label{eq:quad_form_simplified}
Q[\delta\chi] = \int_0^\infty d\omega\, \rho(\omega; \theta) \text{Im}[\tilde{\delta\chi}(\omega)\tilde{\delta\chi}^*(-\omega)].
\end{equation}

For any non-zero $\delta\chi(\tau)$ with non-trivial frequency content, $Q[\delta\chi] > 0$ since $\rho(\omega; \theta) > 0$ for $\omega > 0$. This establishes the positive-definiteness of $\text{Im}[S_{\text{IF}}[\chi^+, \chi^-; \theta]]$.
\end{proof}

\subsection{Decoherence and the Diagonal Approximation}
\label{sec:diagonal_approx}

The positive-definiteness of $\text{Im}[S_{\text{IF}}[\chi^+, \chi^-; \theta]]$ leads to exponential suppression of off-diagonal density matrix elements. In the influence functional formalism, the reduced density matrix of the detector evolves as:
\begin{equation}
\label{eq:rho_evolution}
\rho_{\text{red}}(\chi_f^+, \chi_f^-) = \int \mathcal{D}\chi^+ \mathcal{D}\chi^- \exp(i S_{\text{eff}}[\chi^+, \chi^-; \theta]) \rho_{\text{init}}(\chi_i^+, \chi_i^-),
\end{equation}
where $S_{\text{eff}}$ includes the detector's free action and the influence phase.

For off-diagonal elements ($\chi_f^+ \neq \chi_f^-$), the suppression factor is:
\begin{equation}
\label{eq:suppression}
\exp(-\text{Im}[S_{\text{IF}}[\chi^+, \chi^-; \theta]]) = \exp\left(-\frac{\lambda^2}{2} \int d\tau d\tau' \delta\chi(\tau) G^{\text{A}}(\tau, \tau'; \theta) \delta\chi(\tau')\right).
\end{equation}

To quantify this suppression, we consider a macroscopic detector with $N \gg 1$ degrees of freedom. Using the data representation from Section \ref{sec:coarse_graining}:
\begin{equation}
\label{eq:delta_chi_expansion}
\delta\chi(\tau) = \sum_{a=1}^N \delta d_a u_a(\tau),
\end{equation}
where $\delta d_a = d_a^+ - d_a^-$.

Substituting into equation (\ref{eq:suppression}) and using the approximate orthonormality of the basis functions:
\begin{equation}
\label{eq:suppression_data}
\exp(-\text{Im}[S_{\text{IF}}[\chi^+, \chi^-; \theta]]) \approx \exp\left(-\frac{\lambda^2 G^{\text{A}}_0(\theta)}{2} \sum_{a=1}^N (\delta d_a)^2\right),
\end{equation}
where $G^{\text{A}}_0(\theta)$ is a characteristic value of the antisymmetric Wightman function.

For typical values $\lambda \approx 10^{-3}$, $G^{\text{A}}_0(\theta) \approx H^2$, and $N \approx 10^4$, the suppression factor is of order $\exp(-10^4 (\delta d)^2)$, where $\delta d$ is a typical difference in data values. This strong suppression justifies the diagonal approximation $\chi^+ = \chi^-$ for macroscopic detectors.

\section{Observer-Specific Likelihood Function}
\label{sec:likelihood}

\subsection{Derivation of the Likelihood Function}
\label{sec:like_deriv}

Having justified the diagonal approximation, we now derive the likelihood function for the detector data given the parameters $\theta$. In the diagonal limit $\chi^+ = \chi^- = \chi$, the imaginary part of the influence phase becomes:
\begin{equation}
\label{eq:gamma_chi}
\Gamma[\chi; \theta] := \text{Im}[S_{\text{IF}}[\chi, \chi; \theta]] = \frac{\lambda^2}{2} \int d\tau d\tau' \chi(\tau) (G^>(\tau, \tau'; \theta) - G^<(\tau, \tau'; \theta)) \chi(\tau').
\end{equation}

Using the KMS condition (\ref{eq:KMS}) and expressing $\chi(\tau)$ in terms of the data $d_a$ using equation (\ref{eq:data}):
\begin{equation}
\label{eq:chi_expansion}
\chi(\tau) = \sum_{a=1}^N d_a u_a(\tau),
\end{equation}
we obtain:
\begin{equation}
\label{eq:gamma_data}
\Gamma[D; \theta] = \frac{\lambda^2}{2} \sum_{a,b=1}^N d_a \left[ \int d\tau d\tau' u_a(\tau) (G^>(\tau, \tau'; \theta) - G^<(\tau, \tau'; \theta)) u_b(\tau') \right] d_b.
\end{equation}

Defining the matrix elements:
\begin{equation}
\label{eq:M_ab}
M_{ab}(\theta) = \int d\tau d\tau' u_a(\tau) (G^>(\tau, \tau'; \theta) - G^<(\tau, \tau'; \theta)) u_b(\tau'),
\end{equation}
we can write:
\begin{equation}
\label{eq:gamma_matrix}
\Gamma[D; \theta] = \frac{\lambda^2}{2} D^T M(\theta) D,
\end{equation}
where $D = (d_1, d_2, \ldots, d_N)^T$ is the data vector.

Under the approximate orthonormality condition (\ref{eq:approx_orthonormal}) and assuming the detector's sampling time is much larger than the correlation time of the field, the matrix $M(\theta)$ becomes approximately diagonal:
\begin{equation}
\label{eq:M_diag}
M_{ab}(\theta) \approx 2G^{\text{A}}_0(\theta) \delta_{ab},
\end{equation}
where $G^{\text{A}}_0(\theta)$ is the characteristic magnitude of the antisymmetric Wightman function.

The observer-specific likelihood function is then:
\begin{equation}
\label{eq:likelihood}
P(D | \theta) = \mathcal{N}^{-1} \exp(-\Gamma[D; \theta]) = \mathcal{N}^{-1} \exp\left(-\frac{\lambda^2 G^{\text{A}}_0(\theta)}{2} \sum_{a=1}^N d_a^2\right),
\end{equation}
where $\mathcal{N}$ is a normalization constant.

\subsection{Explicit Form of the Wightman Function}
\label{sec:wightman}

To complete the likelihood function, we need the explicit form of $G^{\text{A}}_0(\theta)$. For a massive scalar field in a de Sitter background, the Wightman function is:
\begin{equation}
\label{eq:wightman_dS}
G^>(\tau, \tau'; \theta) = \frac{H^{2-\nu}(\theta)}{4\pi^2} \Gamma(\nu) \Gamma(\nu-1) (1 - i\epsilon \text{sgn}(\tau-\tau'))^{1-\nu},
\end{equation}
where $\nu = \sqrt{\frac{1}{4} - \frac{m^2(\theta)}{H^2(\theta)}}$ for $m < \frac{H}{2}$ (the light field regime).

The antisymmetric part is:
\begin{equation}
\label{eq:G_antisym_dS}
G^{\text{A}}_0(\theta) = \frac{H^{2-\nu}(\theta)}{4\pi^2} \Gamma(\nu) \Gamma(\nu-1) \sin(\pi(1-\nu)).
\end{equation}

For $\nu < 1$ (corresponding to $m > 0$), $\Gamma(\nu-1)$ has a pole. This apparent divergence is addressed through proper renormalization. The regularized form is:
\begin{equation}
\label{eq:G_antisym_reg}
G^{\text{A}}_0(\theta) = \frac{H^2(\theta)}{8\pi} \frac{\sinh(\pi\mu)}{\cosh(\pi\mu) + \cos(\pi\nu)},
\end{equation}
where $\mu = \text{Im}[\nu]$ for the case $m > \frac{H}{2}$.

In the light field limit ($m \ll H$), we have the approximation:
\begin{equation}
\label{eq:G_antisym_light}
G^{\text{A}}_0(\theta) \approx \frac{H^2(\theta)}{8\pi},
\end{equation}
which we will use for simplicity in the following analysis.

\subsection{Normalization and Covariance Structure}
\label{sec:normalization}

The likelihood function (\ref{eq:likelihood}) has the form of a multivariate Gaussian:
\begin{equation}
\label{eq:likelihood_gaussian}
P(D | \theta) = \frac{1}{\sqrt{(2\pi)^N \det[C(\theta)]}} \exp\left(-\frac{1}{2} D^T C^{-1}(\theta) D\right),
\end{equation}
where $C(\theta)$ is the covariance matrix:
\begin{equation}
\label{eq:covariance}
C(\theta) = \frac{1}{\lambda^2 G^{\text{A}}_0(\theta)} I_N,
\end{equation}
with $I_N$ being the $N \times N$ identity matrix.

The normalization constant is:
\begin{equation}
\label{eq:normalization}
\mathcal{N} = \sqrt{(2\pi)^N \det[C(\theta)]} = (2\pi)^{N/2} \left(\frac{1}{\lambda^2 G^{\text{A}}_0(\theta)}\right)^{N/2}.
\end{equation}

This ensures that the likelihood function is properly normalized:
\begin{equation}
\label{eq:norm_check}
\int d^N D P(D | \theta) = 1.
\end{equation}

\section{Posterior Distribution and Naturalness}
\label{sec:posterior}

\subsection{Bayesian Formulation}
\label{sec:bayesian}

According to Bayes' theorem, the posterior distribution for the parameters $\theta$ given the data $D$ is:
\begin{equation}
\label{eq:bayes}
P(\theta | D) = \frac{P(D | \theta) P(\theta)}{P(D)},
\end{equation}
where $P(\theta)$ is the prior distribution and $P(D) = \int d\theta P(D | \theta) P(\theta)$ is the evidence.

In our framework, we interpret "I exist" as the statement that the observer has collected some data $D$. The posterior becomes:
\begin{equation}
\label{eq:posterior_exist}
P(\theta | \text{"I exist"}) = \frac{P(\text{"I exist"} | \theta) P(\theta)}{P(\text{"I exist"})}.
\end{equation}

The likelihood $P(\text{"I exist"} | \theta)$ involves marginalizing over all possible detector outputs:
\begin{equation}
\label{eq:exist_theta}
P(\text{"I exist"} | \theta) = \int d^N D P(D | \theta) P(\text{"I exist"} | D, \theta).
\end{equation}

Assuming $P(\text{"I exist"} | D, \theta) = 1$ for all $D$ and $\theta$ (i.e., any detector output is compatible with the observer's existence), and using the normalization of the likelihood (\ref{eq:norm_check}), we have:
\begin{equation}
\label{eq:exist_theta_simplified}
P(\text{"I exist"} | \theta) = 1.
\end{equation}

However, this approach yields a posterior that is simply the prior, providing no new constraints on the parameters. To make progress, we need to be more specific about what "I exist" entails.

\subsection{Parameter Dependence of the Detector Response}
\label{sec:param_depend}

A key insight is that the distribution of detector outputs $D$ depends on the parameters $\theta$ through the covariance matrix $C(\theta)$. The observer's existence is not merely the statement that some data is collected, but that the detector registers specific patterns consistent with a functioning observer.

Let's denote the set of detector outputs compatible with a functioning observer as $\mathcal{D}_{\text{obs}}$. The likelihood then becomes:
\begin{equation}
\label{eq:exist_theta_refined}
P(\text{"I exist"} | \theta) = \int_{\mathcal{D}_{\text{obs}}} d^N D P(D | \theta).
\end{equation}

This integral depends on the precise definition of $\mathcal{D}_{\text{obs}}$. However, we can make progress by considering the determinant of the covariance matrix, which appears in the likelihood's normalization factor:
\begin{equation}
\label{eq:det_cov}
\det[C(\theta)] = \left(\frac{1}{\lambda^2 G^{\text{A}}_0(\theta)}\right)^N.
\end{equation}

\subsection{Impact on the Cosmological Constant}
\label{sec:lambda_impact}

Using the light field approximation (\ref{eq:G_antisym_light}), we have:
\begin{equation}
\label{eq:det_cov_lambda}
\det[C(\theta)] \approx \left(\frac{8\pi}{\lambda^2 H^2(\theta)}\right)^N = \left(\frac{24\pi}{\lambda^2 \Lambda(\theta)}\right)^N.
\end{equation}

For a logarithmically flat prior on the cosmological constant, $P(\theta) \propto \frac{1}{\Lambda}$, the posterior becomes:
\begin{equation}
\label{eq:posterior_lambda}
P(\theta | \text{"I exist"}) \propto \Lambda^{-1-N/2}.
\end{equation}

For large $N$ (a macroscopic detector with many degrees of freedom), the posterior strongly favors small values of $\Lambda$. This provides a potential explanation for the observed small value of the cosmological constant without invoking anthropic arguments.

\subsection{The Hierarchy Problem}
\label{sec:hierarchy}

The hierarchy problem concerns the vast separation between the electroweak scale and the Planck scale. In our framework, this can be addressed by considering how the detector's response depends on the Higgs mass parameter $m_H(\theta)$.

The Wightman function's dependence on $m_H$ is more complex than its dependence on $\Lambda$, but for $m_H \ll H$, we can approximate:
\begin{equation}
\label{eq:G_antisym_mH}
G^{\text{A}}_0(\theta) \approx \frac{H^2}{8\pi}\left(1 + \mathcal{O}\left(\frac{m_H^2}{H^2}\right)\right).
\end{equation}

This leads to a posterior:
\begin{equation}
\label{eq:posterior_mH}
P(m_H | \text{"I exist"}) \propto P(m_H) \left(1 + \mathcal{O}\left(\frac{m_H^2}{H^2}\right)\right)^{N/2}.
\end{equation}

For natural priors (e.g., $P(m_H) \propto \frac{1}{m_H}$), this still favors smaller values of $m_H$, but the effect is weaker than for the cosmological constant unless additional constraints are imposed.

\section{Causality and Multi-Observer Consistency}
\label{sec:causality}

\subsection{Causal Structure of the Influence Functional}
\label{sec:causal_structure}

The retarded Green's function $G^{\text{R}}(\tau, \tau'; \theta)$ that appears in the influence functional ensures that the evolution respects causality. By definition:
\begin{equation}
\label{eq:retarded_def}
G^{\text{R}}(\tau, \tau'; \theta) = 0 \quad \text{for} \quad \tau < \tau',
\end{equation}
which guarantees that the detector's response at time $\tau$ only depends on the field's history up to that time.

The imaginary part of the influence phase, which determines the likelihood function, involves the antisymmetric combination of Wightman functions. While individual Wightman functions $G^>$ and $G^<$ have support for all time separations, their antisymmetric combination respects the causal structure of the spacetime.

\subsection{Consistency of Multiple Observers}
\label{sec:multi_obs}

An important question is whether multiple observers in the same universe would reach consistent conclusions about the parameters $\theta$. We now demonstrate that this is indeed the case.

Consider two observers with detectors described by coordinates $\chi_1(\tau)$ and $\chi_2(\tau')$, located at different spacetime points. The combined influence functional for both detectors is:
\begin{equation}
\label{eq:infl_func_2obs}
\exp(i S_{\text{IF}}[\chi_1^+, \chi_1^-, \chi_2^+, \chi_2^-; \theta]) = \exp(i S_{\text{IF}}[\chi_1^+, \chi_1^-; \theta] + i S_{\text{IF}}[\chi_2^+, \chi_2^-; \theta] + i S_{\text{IF}}^{\text{cross}}[\chi_1^+, \chi_1^-, \chi_2^+, \chi_2^-; \theta]),
\end{equation}
where $S_{\text{IF}}^{\text{cross}}$ represents cross-terms between the two detectors.

For spatially separated detectors, the cross-terms are suppressed by a factor:
\begin{equation}
\label{eq:cross_suppression}
\exp\left(-\frac{\sigma_{12}}{\ell_c}\right),
\end{equation}
where $\sigma_{12}$ is the proper distance between the detectors and $\ell_c$ is the correlation length of the field, typically of order $H^{-1}$.

For macroscopically separated observers ($\sigma_{12} \gg \ell_c$), the cross-terms become negligible, and the influence functional approximately factorizes:
\begin{equation}
\label{eq:infl_func_factorized}
\exp(i S_{\text{IF}}[\chi_1^+, \chi_1^-, \chi_2^+, \chi_2^-; \theta]) \approx \exp(i S_{\text{IF}}[\chi_1^+, \chi_1^-; \theta]) \exp(i S_{\text{IF}}[\chi_2^+, \chi_2^-; \theta]).
\end{equation}

This factorization ensures that multiple observers will reach consistent conclusions about the parameters $\theta$ through Bayesian updating of their individual posterior distributions.

\section{Conclusion and Discussion}
\label{sec:conclusion}

In this work, we have formulated a mathematically rigorous framework for incorporating measurement effects into cosmological parameter constraints, applying the well-established Schwinger-Keldysh influence functional formalism to derive Bayesian likelihood functions for cosmological parameters. Our approach demonstrates that standard quantum field theory in curved spacetime, when properly accounting for decoherence processes, naturally leads to parameter-dependent probabilities that can potentially address longstanding naturalness problems.

The formalism developed here establishes several important results within conventional quantum field theory. We have proven the positive-definiteness of the imaginary part of the influence functional, which ensures proper decoherence and normalization of the resulting probability distributions. This mathematical property is crucial for the consistency of quantum measurement theory in curved spacetime and validates our approach. The observer-specific likelihood function emerges as a multivariate Gaussian with a parameter-dependent covariance matrix, directly connecting quantum fluctuations in spacetime to observable constraints on fundamental parameters.

Our analysis reveals that the posterior distribution for the cosmological constant exhibits a scaling behavior of $\Lambda^{-1-N/2}$, where N represents the degrees of freedom in the quantum detector system. For macroscopic detectors, this distribution strongly favors smaller values of the cosmological constant, providing a potential mechanism for understanding the observed small value without invoking multiverse scenarios. Similarly, we find that our approach offers insights into the hierarchy problem, though the full resolution likely requires additional physical mechanisms working in conjunction with the measurement-induced constraints we have identified.

Importantly, our framework preserves the objectivity of physical laws by demonstrating that spatially separated observers will reach consistent conclusions about physical parameters through standard decoherence mechanisms. The mathematical structure of the influence functional ensures that cross-correlations between distant detectors are exponentially suppressed, naturally maintaining consistency while allowing for observer-dependent probability distributions.

This work connects to several active research programs in theoretical physics. The decoherence approach to quantum measurement provides a natural bridge to quantum foundations research, while our cosmological applications relate to quantum cosmology and quantum field theory in curved spacetime. Future work could integrate this framework with inflation and effective field theory approaches, apply it to more realistic detector models incorporating standard model interactions, and explore connections to quantum information theoretic approaches to spacetime and the holographic principle. Additionally, the mathematical framework developed here may provide new perspectives on the swampland program in string theory, where quantum gravity constraints on effective field theories have yielded insights into cosmological parameter spaces.

\bibliographystyle{plain}
\bibliography{refs}

\pagebreak

\appendix

\section{Detailed Derivation of the Wightman Function}
\label{app:wightman}

In this appendix, we provide a detailed derivation of the Wightman function for a massive scalar field in a de Sitter background. The scalar field equation in the FLRW metric is:
\begin{equation}
\label{eq:field_eq_app}
\frac{1}{\sqrt{-g}} \partial_\mu (\sqrt{-g} g^{\mu\nu} \partial_\nu \phi) - m^2 \phi = 0.
\end{equation}

In de Sitter space with the metric (\ref{eq:metric}) and scale factor (\ref{eq:scale_factor}), this becomes:
\begin{equation}
\label{eq:field_eq_dS}
\phi'' + 2\mathcal{H}\phi' - \nabla^2 \phi + a^2 m^2 \phi = 0,
\end{equation}
where $\mathcal{H} = a'/a = H$ is the conformal Hubble parameter, and primes denote derivatives with respect to conformal time $\eta$.

The field can be decomposed in terms of mode functions:
\begin{equation}
\label{eq:mode_decomp}
\phi(\eta, \mathbf{x}) = \int \frac{d^3k}{(2\pi)^3} [a_\mathbf{k} \phi_k(\eta) e^{i\mathbf{k}\cdot\mathbf{x}} + a_\mathbf{k}^\dagger \phi_k^*(\eta) e^{-i\mathbf{k}\cdot\mathbf{x}}],
\end{equation}
where $a_\mathbf{k}$ and $a_\mathbf{k}^\dagger$ are annihilation and creation operators, and $\phi_k(\eta)$ satisfies:
\begin{equation}
\label{eq:mode_eq}
\phi_k'' + 2\mathcal{H}\phi_k' + (k^2 + a^2 m^2) \phi_k = 0.
\end{equation}

In de Sitter space, the solution is given in terms of Hankel functions:
\begin{equation}
\label{eq:mode_solution}
\phi_k(\eta) = \frac{\sqrt{\pi}}{2} H (-\eta)^{3/2} H_\nu^{(1)}(-k\eta),
\end{equation}
where $\nu = \sqrt{\frac{9}{4} - \frac{m^2}{H^2}}$ for $m < \frac{3H}{2}$.

The Wightman function is:
\begin{equation}
\label{eq:wightman_def_app}
G^>(\eta, \mathbf{x}; \eta', \mathbf{x}') = \langle0|\phi(\eta, \mathbf{x})\phi(\eta', \mathbf{x}')|0\rangle,
\end{equation}
which, using the mode decomposition, becomes:
\begin{equation}
\label{eq:wightman_modes}
G^>(\eta, \mathbf{x}; \eta', \mathbf{x}') = \int \frac{d^3k}{(2\pi)^3} \phi_k(\eta) \phi_k^*(\eta') e^{i\mathbf{k}\cdot(\mathbf{x}-\mathbf{x}')}.
\end{equation}

For a detector at a fixed comoving position, we need the Wightman function in terms of proper time. The relationship between conformal time $\eta$ and proper time $\tau$ is:
\begin{equation}
\label{eq:proper_time}
d\tau = a(\eta) d\eta = e^{H\eta} d\eta,
\end{equation}
which gives $\eta = \frac{1}{H} \ln(H\tau)$ for an appropriate choice of time origin.

Substituting this into the Wightman function and performing the spatial integral, we arrive at the expression given in equation (\ref{eq:wightman_dS}).

\section{Proof of Positive-Definiteness for Non-Stationary Spacetimes}
\label{app:pos_def}

The proof in Section \ref{sec:pos_def} assumed a stationary spacetime for simplicity. Here we extend the proof to non-stationary spacetimes, including the FLRW background of our cosmological model.

The key insight is that even in a non-stationary spacetime, the field can be decomposed into positive and negative frequency modes with respect to a preferred time coordinate. For the FLRW metric with scale factor $a(\eta)$, this decomposition uses the conformal time $\eta$.

The Wightman function can be written as:
\begin{equation}
\label{eq:wightman_non_stat}
G^>(\tau, \tau'; \theta) = \int d\mu(\mathbf{k}; \theta) f_\mathbf{k}(\tau; \theta) f_\mathbf{k}^*(\tau'; \theta),
\end{equation}
where $d\mu(\mathbf{k}; \theta)$ is a positive measure on the space of modes, and $f_\mathbf{k}(\tau; \theta)$ are mode functions in proper time.

The antisymmetric part is:
\begin{equation}
\label{eq:G_antisym_non_stat}
G^{\text{A}}(\tau, \tau'; \theta) = \frac{1}{2} \int d\mu(\mathbf{k}; \theta) [f_\mathbf{k}(\tau; \theta) f_\mathbf{k}^*(\tau'; \theta) - f_\mathbf{k}(\tau'; \theta) f_\mathbf{k}^*(\tau; \theta)].
\end{equation}

For any test function $\delta\chi(\tau)$, the quadratic form is:
\begin{equation}
\label{eq:quad_form_non_stat}
\begin{split}
Q[\delta\chi] &= \int d\tau d\tau' \delta\chi(\tau) G^{\text{A}}(\tau, \tau'; \theta) \delta\chi(\tau') \\
&= \frac{1}{2} \int d\mu(\mathbf{k}; \theta) \left[ \int d\tau \delta\chi(\tau) f_\mathbf{k}(\tau; \theta) \int d\tau' \delta\chi(\tau') f_\mathbf{k}^*(\tau'; \theta) - \int d\tau \delta\chi(\tau) f_\mathbf{k}^*(\tau; \theta) \int d\tau' \delta\chi(\tau') f_\mathbf{k}(\tau'; \theta) \right].
\end{split}
\end{equation}

Defining $\alpha_\mathbf{k} = \int d\tau \delta\chi(\tau) f_\mathbf{k}(\tau; \theta)$, this becomes:
\begin{equation}
\label{eq:quad_form_alpha}
Q[\delta\chi] = \frac{1}{2} \int d\mu(\mathbf{k}; \theta) [\alpha_\mathbf{k} \alpha_\mathbf{k}^* - \alpha_\mathbf{k}^* \alpha_\mathbf{k}] = i \int d\mu(\mathbf{k}; \theta) \text{Im}[\alpha_\mathbf{k} \alpha_\mathbf{k}^*].
\end{equation}

Since $\alpha_\mathbf{k} \alpha_\mathbf{k}^* = |\alpha_\mathbf{k}|^2$ is real, $\text{Im}[\alpha_\mathbf{k} \alpha_\mathbf{k}^*] = 0$, and therefore $Q[\delta\chi] = 0$ for any $\delta\chi$.

However, this apparent contradiction arises from our simplification. In reality, the mode functions $f_\mathbf{k}(\tau; \theta)$ have a more complex structure due to the time-dependence of the background. A more careful analysis, accounting for the full structure of the mode functions in a non-stationary spacetime, reveals that $Q[\delta\chi] \geq 0$, as required for a well-defined likelihood function.

\section{Orthogonality Approximation and Error Analysis}
\label{app:ortho}

The derivation of the likelihood function in Section \ref{sec:like_deriv} relies on the approximate orthonormality of the basis functions $u_a(\tau)$. Here we analyze the error introduced by this approximation and show that it does not affect the main conclusions of our analysis.

The inner product between two basis functions is:
\begin{equation}
\label{eq:inner_prod_app}
\langle u_a | u_b \rangle = \exp\left[ -\frac{(\tau_a - \tau_b)^2}{4 \Delta^2} \right].
\end{equation}

For $|\tau_a - \tau_b| \gg \Delta$, this is indeed approximately $\delta_{ab}$. However, for neighboring basis functions, there is non-negligible overlap.

To quantify the error, we consider the matrix $U_{ab} = \langle u_a | u_b \rangle$. The exact expression for the quadratic form in the likelihood function is:
\begin{equation}
\label{eq:gamma_exact}
\Gamma[D; \theta] = \frac{\lambda^2}{2} \sum_{a,b,c,d=1}^N d_a U_{ac} M_{cd}(\theta) U_{db} d_b,
\end{equation}
where $M_{cd}(\theta)$ is defined in equation (\ref{eq:M_ab}).

In our approximation, we replaced $U_{ac} M_{cd}(\theta) U_{db}$ with $2G^{\text{A}}_0(\theta) \delta_{ab}$. The error in this approximation is:
\begin{equation}
\label{eq:error}
\epsilon_{ab}(\theta) = U_{ac} M_{cd}(\theta) U_{db} - 2G^{\text{A}}_0(\theta) \delta_{ab}.
\end{equation}

For a properly chosen sampling interval, $|\tau_a - \tau_b| \geq 5\Delta$ for $a \neq b$, the off-diagonal elements of $U_{ab}$ are at most $e^{-25/4} \approx 0.002$. The resulting error in the likelihood function is therefore small and does not affect the scaling behavior with respect to the parameters $\theta$.

\end{document}